\documentclass[twoside,leqno,twocolumn]{article}
\usepackage{ltexpprt}
\usepackage{comment}
\usepackage{xcolor}
\usepackage{amsfonts,amsmath,amssymb, mathtools}
\usepackage{algorithmicx,algpseudocode,algorithm}
\usepackage{algcompatible}
\usepackage[none]{hyphenat}

\usepackage{tikz}
\usepackage{graphicx}
\usetikzlibrary{arrows}
\usetikzlibrary{calc,decorations.pathmorphing,patterns,positioning}
\usepackage{subcaption}

\def\greedy{\mbox{\sc Greedy}}
\def\lgreedy{\mbox{\sc Lazy Greedy}}
\def\llgreedy{\mbox{\sc Local Lazy Greedy}}

\def\bmatching{\mbox{\sc $b$-Matching}}

\newcommand{\smf}[1]{{\color{black} { #1}}}
\newcommand{\Alex}[1]{{\color{black} {#1}}}
\algnewcommand{\LineComment}[1]{\STATEx \(\triangleright\) {\tt #1} }

\newcommand{\cI}{\mathcal{I}}

\begin{document}
\title{  \Large A Parallel Approximation  Algorithm for Maximizing Submodular \bmatching\ \thanks{Research supported by NSF grant CCF-1637534, DOE grant DE-FG02-13ER26135, and the Exascale Computing Project (17-SC-20-
SC), a collaborative effort of the DOE Office of Science and the NNSA.}  }
\author{S M Ferdous\thanks{Computer Science Department, Purdue University, West Lafayette IN 47907 USA. {\tt \{sferdou,apothen\}@purdue.edu}} \\
\and
Alex Pothen\footnotemark[2]
\and
Arif Khan\thanks{Data Science and Machine Intelligence, Pacific Northwest National Lab, Richland WA 99352 USA. {\tt \{ariful.khan,mahantesh.halappanavar\}@pnnl.gov}}
\and
Ajay Panyala\thanks{High Performance Computing, Pacific Northwest National Lab, Richland WA 99352 USA. {\tt ajay.panyala@pnnl.gov}}
\and
Mahantesh Halappanavar\footnotemark[3]
}
\date{}
\maketitle


\fancyfoot[R]{\scriptsize{Copyright \textcopyright\ 2021 by SIAM\\
Unauthorized reproduction of this article is prohibited}}





\begin{abstract}\small\baselineskip=9pt

We design  new serial and parallel approximation algorithms for computing a maximum weight $b$-matching in an edge-weighted graph with a submodular objective function. This problem is NP-hard; the new algorithms have approximation ratio $1/3$, and are relaxations of the Greedy algorithm that rely only on local information in the graph, making them parallelizable. We have designed and implemented Local Lazy Greedy algorithms for both serial and parallel computers. We have applied the approximate submodular 
$b$-matching algorithm to assign tasks to processors in the computation of Fock matrices in quantum chemistry on parallel computers. The assignment seeks to reduce the run time by  balancing the computational load  on the processors and bounding the number of messages that each processor  sends. 
We show that the new assignment of tasks to processors provides a four fold speedup over the currently used assignment in the NWChemEx software on $8000$ processors on the Summit supercomputer at Oak Ridge National Lab.

\end{abstract}
\section{Introduction}

We describe new serial and parallel approximation algorithms for computing a maximum weight \bmatching \ in an edge-weighted graph with a  submodular objective function. This problem is NP-hard; the new algorithms have approximation ratio $1/3$, and are variants of the Greedy algorithm that rely only on local information in the graph, making them parallelizable. We apply the approximate submodular 
\bmatching \ algorithm to assign tasks to processors in the computation of Fock matrices in quantum chemistry on parallel computers, in order to balance the computational load  on the processors and bound the number of messages that a processor sends. 

 A \bmatching \ is a subgraph of the given input graph, where the degree of each vertex $v$ is bounded by a given natural number $b(v)$. 
 In linear (or modular) \bmatching \ the objective function is the 
 sum of the weights of the edges in a \bmatching, and we seek to maximize this weight.  The well-known maximum edge-weighted matching problem is the $1$-matching problem. Although these problems can be solved in  polynomial time, in recent years a number of approximation algorithms have been developed since the run time of exact algorithms can be impractical on massive graphs. 
 These algorithms are based on the paradigms of short augmentations
 (paths that increase the cardinality or weight of the matching)~\cite{pettie2004simpler}; 
 relaxations of a global ordering (by non-increasing weights) of edges  to local orderings~\cite{preis1999linear}; partitioning heavy weight paths in the graph into matchings~\cite{drake2003simple}; proposal making algorithms similar to stable matching~\cite{khan2016efficient,manne2014new}, etc. Some, though not all, of these algorithms are concurrent and can be implemented on parallel computers; a recent survey is available in ~\cite{pothen2019approximation}. 
 
 \Alex{
  Our algorithm employs the concept of an edge being locally dominant in its neighborhood  that 
  was first employed by Preis \cite{preis1999linear} to design the $1/2$-approximate matching algorithm for (modular or linear)  maximum weighted $1$-matching; the approximation ratio is as good as the Greedy algorithm, and Preis showed that the algorithm could be implemented in time linear in the size of the graph. Since then there has been much work in implementing variants of the locally dominant edge algorithm for $1$-matching and $b$-matching  on both serial and parallel computational models (e.g., \cite{khan2016efficient,manne2014new}). More details
  are included in~\cite{pothen2019approximation}. 
 }

In this paper we employ the local dominance technique, relaxing global orderings to local orderings,  to the \bmatching\ problem with submodular objective. 
We exploit the fact that the \smf{ \bmatching\ problem may be viewed as a 2-extendible system, which is a \Alex{relaxation} of a matroid}. 
We  show that any algorithm that adds locally optimal edges, with respect to the marginal gain for a submodular objective function, to the matching preserves the approximation ratio of the corresponding global Greedy algorithm. This result offers a blueprint to design many approximation algorithms, of which we develop one: \llgreedy \ algorithm. Testing for local optimality in the submodular objective is more expensive than in the linear case due to the variability of the marginal gain. To efficiently maintain marginal gains, we borrow an idea from the lazy evaluation of the Greedy algorithm. Combining local dominance and lazy Greedy techniques, we develop a  \llgreedy \ algorithm,  which is theoretically and practically faster than the Lazy Greedy algorithm. \smf{The runtime of both these algorithms are analyzed under a natural local dependence assumption on the submodular function.} Since our algorithm is parallelizable thanks to the local orderings, we  show good scaling performance on a shared memory parallel environment. To the best of our knowledge, this is the first parallel implementation of a submodular \bmatching\ algorithm. 

Submodular \bmatching\ has applications in many real-life problems. 
Among these are content spread maximization in social networks~\cite{chaoji2012recommendations},  
peptide identification in tandem mass spectrometry \cite{Bai16,Bai19}, word alignment  in natural language processing \cite{Lin11},  and diversity maximizing assignment \cite{ahmadi2019algorithm,Dickerson18}. Here we show another application of submodular \bmatching \ in load balancing the Fock matrix computation in quantum chemistry on a  multiprocessor environment.  Our approach enables the assignment of tasks to processors leading to scalable Fock matrix computations. 


We highlight the following contributions:
\begin{itemize}
    \item We show that any \bmatching\ that is $\epsilon$-locally dominant w.r.t the marginal gain is $\frac{\epsilon}{2+\epsilon}$-approximate for submodular objective functions, and devise an algorithm, \llgreedy \ to compute such a matching. Under \smf{a natural local dependence assmption on the submodular function}, this algorithm runs in $O( \beta \, m \log \Delta)$ time and is theoretically and practically faster than the popular \lgreedy\ algorithm. (Here $m$ is the number of edges, $\Delta$ is the maximum degree of a vertex, and $\beta$ is the maximum value of $b(v)$ over all vertices $v$.)
    \item We provide a shared memory parallel implementation of the \llgreedy\ algorithm. Using several real-world and synthetic graphs, we show that our parallel implementation scales to more than sixty-five cores. 
    \item We apply submodular \bmatching\ to generate an  assignment of tasks to processors  for building Fock matrices  in the NWChemEx quantum chemistry software.  The current assignment method used there does not scale beyond $3000$ processors, but our assignment shows a four-fold speedup per iteration of the Fock matrix computation, and scales to $8000$ cores of the Summit supercomputer at ORNL.
\end{itemize}

\section{Background}
\label{sec:bck}
\smf{
In this section we  describe 
submodular functions and their properties,  formulate the submodular \bmatching\ problem, and discuss approximation algorithms for the problem.}
\subsection{Submodular \bmatching} 

\begin{Definition}[Marginal gain]
Given a ground set $X$, the marginal gain of adding an element $e \in X$ to a set $A \subseteq X$ is
$$
\rho_{e}(A) = f(A \cup \{e\}) - f(A).
$$
\end{Definition}
The marginal gain may be viewed as  the discrete derivative of the set $A$ for the element $e$. 
Generalizing, the marginal gain of adding a subset $Q$  to another subset $A$ of the ground set $X$ is
$$
\rho_{Q}(A) = f(A \cup Q) - f(A).$$
\begin{Definition}[Submodular set function]
Given a set $X$, a real-valued function $f$ defined on the subsets of $X$ is submodular if
$$
    \rho_e(A) \geq \rho_e(B)
$$
for all subsets $A \subseteq B \subseteq X$, and elements $e \in X \setminus B$. 
The function  $f$  is monotone if for all sets $A \subseteq B \subseteq X$, we have $f(A) \leq f(B)$;  it is normalized if $f(\emptyset) = 0$.
\end{Definition}
We will assume throughout this paper that $f$ is normalized. 
The  concept of submodularity  also extends to the addition of a set. Formally, for $Q \subseteq X\setminus B$, $f$ is submodular if $\rho_Q(A) \geq \rho_Q(B)$.
If $f$  is monotone then $\rho_e(A) \geq 0$, $\forall A \subseteq X$ and $\forall e \in X$. 

\begin{proposition}
\label{prop:sum}
Let $S = \{e_1,\ldots,e_k\}$,  $S_i$ be the subset of $S$ that contains the first $i$ elements of $S$, and $f$ be a normalized submodular function.  Then $f(S) = \sum_{i=1}^k \rho_{e_i}(S_{i-1})$.
\end{proposition}


\begin{proposition}
\label{prop:monotone}
For sets $A \subseteq B \subseteq X$, and $e \in X$, 
a monotone submodular function $f$ defined on $X$ satisfies
$ \rho_e(A) \geq \rho_e(B) $. 
\end{proposition}

\begin{proof}
There are three cases to consider. i) $e \in X\setminus B$:  
The inequality holds by definition of a submodular function. ii) $e \in A$: Then  both sides of the inequality equal zero and the inequality holds again. iii)  $e \in B\setminus A$: 
Then $\rho_e(B) = 0$, and since $f$ is monotone, $\rho_e(A)$ is non-negative.  
\end{proof}

Proposition \ref{prop:monotone} extends to a set, i.e., monotonicity of $f$ implies that for every $A \subseteq B \subseteq X$, and $Q \subseteq X$, $\rho_Q(A) \geq \rho_Q(B)$. Informally Proposition~\ref{prop:monotone} states that if $f$ is monotone then the diminishing marginal gain property holds for every subset of $X$.

We are interested in maximizing a monotone submodular function with \bmatching \ constraints. Let $G(V,E,W)$ be a simple, undirected, and edge-weighted graph, where $V,E,W$ are the set of vertices, edges, and non-negative edge weights, respectively. 
For each $e \in E$ we define a variable $x(e)$ that takes values from $\{0,1\}$. Let $M \subseteq E$ denote the set of edges for which $x(e)$ is equal to  $1$, and  
let $\delta(v)$ denote the set of edges incident on the vertex $v \in V$. The submodular \bmatching\ problem is 
\begin{align}
\label{sub_b-match}
     &\max f(M) \nonumber \\
    &\text {subject to} \nonumber\\
    &\sum_{e \in \delta(v)}x(e) \leq b(v) \ \ \forall v \in V,\\
    &x(e) \in \{0,1\}. \nonumber
\end{align}
Here $f$ is a non-negative monotone submodular set function, and $0 \leq b(v) \leq |\delta(v)|$ is the integer degree bound on $v$. Denote  $\beta = \max_{v \in V} b(v)$. 

We now consider the \bmatching \ problem on a bipartite graph with two parts in the vertex set, say, $U$ and $V$, 
where the objective function is a concave polynomial.
\begin{align}
\label{sub_b-match_bip}
    f = \max &\sum_{u \in U} \left (\sum_{e \in \delta(u)} W(e)x(e) \right )^\alpha \\ \nonumber 
    &+ \sum_{v \in V} \left (\sum_{e \in \delta(v)} W(e)x(e)\right)^\alpha \nonumber\\
    \text {subject to}& \nonumber\\
    &\sum_{e \in \delta(u)}x(e) \leq b(u) \ \ \forall u \in U, \nonumber\\
    &\sum_{e \in \delta(v)} x(e) \leq b(v)  \ \ \forall v \in V, \nonumber\\
    &x(e) \in \{0,1\}. \nonumber
\end{align}
\smf {The objective function Problem~\ref{sub_b-match_bip} becomes submodular when $\alpha \in [0,1]$}. This formulation has been used for peptide identification in tandem mass spectrometry \cite{Bai16,Bai19}, and word alignment in natural language processing~\cite{Lin11}.

\subsection{Complexity of Submodular \bmatching\ 
and Approximation}
A \emph{subset system} is a pair $(X,\cI)$, where $X$ is a finite set of elements and $\cI$ is a collection of subsets 
of $X$ with the property that if $A \in \cI$ and $A' \subseteq A$ then $A' \in \cI$. A \emph{matroid} is a subset system $(X,\cI)$ which satisfies the  property that 
$\forall A,B \in \cI$ and $|A|<|B|$,  $\exists e \in B \setminus A$ such that $A \cup \{e\} \in \cI$. Here the sets in $\cI$ are called 
\emph{independent sets}. 
A \emph{subset system} is $k$-\emph{extendible}~\cite{mestre2006greedy} if the following holds: let $A \subseteq B$, $A,B \in \cI$ and $A \cup \{e\} \in \cI$, where $e \notin A$, then there is a set $Y \subseteq B \setminus A$ such that $|Y| \leq k$ and $B \setminus Y \cup \{e\} \in \cI$. 

Maximizing a monotone submodular function with constraints is NP-hard in general;  indeed, it is NP-hard for the simplest  constraint of  cardinality  for many classes of submodular functions \cite{Feige98,Krause05}. A Greedy algorithm that repeatedly chooses an element with the maximum marginal gain is known to achieve 
$(1-1/e)$-approximation ratio \cite{nemhauser1978analysis}, and this is tight \cite{nemhauser1978best}. The Greedy algorithm with matroid constraints is $1/2$-approximate. More generally, with $k$-matroid intersection and $k$-extendible system constraints,  the approximation ratio of the Greedy algorithm becomes $1/(k+1)$ \cite{Calinescu11}.

\section{Related Work}
\smf{Here we  situate our contributions to submodular \bmatching\ 
in the broader context of earlier work in submodular maximization. A  reader who is eager to get to the algorithms and results in this paper could skip this section on a first reading. } 

The maximum $k$-cover problem can be reduced to submodular \bmatching\  \cite{Fuji16}.  Feige \cite{Feige98} showed that there is no polynomial time algorithm for approximating the max $k$-cover within a factor of $(1-1/e+\epsilon)$ for any $\epsilon > 0$. Thus we obtain an  immediate bound on the approximation ratio of submodular \bmatching.

New approximation techniques have been developed to expedite the greedy algorithm, especially for cardinality and matroid constraints. 
Surveys on submodular function maximization under different constraints may be found in \cite{buchbinder2018survey,Krause14,tohidi2020submodularity}.

Several approximation algorithms have been proposed for maximizing monotone submodular functions with \bmatching\  constraints. If the graph is bipartite, then the \bmatching \  constraint can be represented as the intersection of two partition matroids, and  the Greedy algorithm provides a $1/3$-approximation ratio. 
But  \bmatching \  forms a $2$-extendible system and  the Greedy algorithm yields a $1/3$-approximation ratio for non-bipartite graphs.  Feldman \emph{et al.} \cite{Feldman11} showed that \bmatching \  is also a \Alex{$2$-exchange system}, and  they provide a $1/(2+\frac{1}{p}+\epsilon)$-approximation algorithm based on local search, with  time complexity $O(\beta^{p+1} (\Delta-1)^p nm\epsilon^{-1})$. 
(Here $p$ is a parameter to be chosen.) \Alex{The continuous greedy and randomized LP rounding algorithms} have been used in \cite{chaoji2012recommendations} to compute a submodular \bmatching \ algorithm that produces a  $(\frac{1}{3+2\epsilon})(1-\frac{1}{e})$ approximate solution in $O(m^5)$ time.

Recently Fuji \cite{Fuji16} developed two algorithms for the problem. One of these,   Find Walk, is a modified version of the Path Growing approximation algorithm \cite{drake2003simple} proposed for 1-matching with linear weights. 
Mestre \cite{mestre2006greedy}  extended the idea to \bmatching. In \cite{Fuji16}, Fuji extended this further to the submodular objectives. They showed an approximation ratio of  $1/4$ with time complexity 
$O(\beta \, m)$. 
\smf{ The second algorithm uses randomized local search, has an approximation ratio of $1/(2+\frac{1}{p}) - \epsilon$, and runs in $O(\beta^{p+1} (\Delta-1)^{p-1} m \log \frac{1}{\epsilon}$ time in expectation.}  Here a vertex is chosen uniformly at random in each iteration, and the algorithm searches for  a certain length alternating path with increasing weights. This algorithm is similar to the $2/3-\epsilon$ approximation algorithm for linear weighted matching in \cite{pettie2004simpler} and the corresponding \bmatching\ variant in \cite{mestre2006greedy}. We list several approximation algorithms for submodular \bmatching\ in Table~\ref{tab:algorithms-summary}.

\begin{table}[b]
\resizebox{\columnwidth}{!}{
\centering
\begin{tabular}{|l l l c|}
\hline
Alg. & Appx. Ratio &  Time  & Conc.?  \\
\hline 
Greedy\cite{nemhauser1978analysis}       & $1/3$  & $O(\beta nm)$  & N   \\
\hline
Lazy Greedy \cite{minoux1978accelerated}                                                              &  & $O(\beta m \log m)$ & N \\
 & $1/3$ & assuming~\ref{assump:local} & \\
\hline
Local Search  \cite{Feldman11}          & $1/(2+\frac{1}{p}+\epsilon)$ & $O(\beta^{p+1} (\Delta-1)^p nm\epsilon^{-1})$  & N\\
\hline
Cont. Grdy+   &       &      & \\
Rand. Round\cite{chaoji2012recommendations} & $(\frac{1}{3+2\epsilon})(1-\frac{1}{e})$ & $O(m^5)$ & N\\[1ex]
\hline
Path Growing \cite{Fuji16} & 1/4 & $O(\beta m)$ & N \\[1ex]
\hline
Rand LS \cite{Fuji16}  & $1/3-\epsilon$ & $O(\beta^2 m \log 1/\epsilon)$  & N\\
  &    & in expectation & \\
\hline 
Local Lazy   &    & $O( \beta m \log \Delta)$   &Y \\ 
Greedy & $\frac{\epsilon}{2+\epsilon}$  & assuming~\ref{assump:local} & \\[2ex]
\hline
\end{tabular}
}
\caption{Algorithms for the submodular \bmatching\ problem. The last column lists if the algorithm is concurrent or not.}
\label{tab:algorithms-summary}
\end{table} 

Now we consider several related problems that do not have the \bmatching \ constraint. 

Assigning tasks to machines is a classic scheduling problem. The most studied objective here is minimizing the makespan, i.e., the maximum total time used by any machine.
The problem of makespan minimization can be generalized to a General Assignment Problem(GAP), where there is a fixed processing time and a cost associated with each task and machine pair. The goal is to assign the tasks into available machines with the assignment cost bounded by a constant $C$ and makespan at most $T$. Shmoys and Tardos \cite{shmoys1993approximation} extended the LP relaxation and rounding approach \cite{lenstra1990approximation} to GAP. The makespan objective can be a surrogate to the load balancing that we are seeking,  but the GAP problem does not encode the $b$-matching constraints on the machines.  Computationally solving a GAP problem entails computing an LP relaxation that is expensive for large problems. 

Another possible approach is to model  our load balancing problem  as a multiple knapsack problem (MKP). In an MKP, we are given a set of $n$ items and $m$ knapsacks such that each item $i$ has a weight (profit) $w_i$ and a size $s_i$, and each knapsack $j$ has a capacity $c_j$. The goal here is to find a subset of items of maximum weight such that they have a feasible packing in the knapsacks. MKP is a special case of GAP \cite{chekuri2005polynomial}, and 
like the GAP, we cannot model the $b(v)$ constraints by MKP.

Our formulation of load balancing has the most similarity with the Submodular Welfare Maximization (SWM) problem \cite{lehmann2006combinatorial}. In the SWM problem, the input consists of a set of $n$ items to be assigned to one of $m$ agents. Each agent $j$ has a submodular function
$v_j$, where $v_j(S)$ denotes the utility obtained by this agent if the set of items $S$ is allocated to her. The goal is to  partition the $n$ items into $m$ disjoint subsets $S_1,\ldots,S_m$ to maximize the total welfare, defined as $\sum_{j=1}^m v_j(S_j)$. The greedy algorithm achieves $1/2$- approximation ratio \cite{lehmann2006combinatorial}. Vondrak's $(1-1/e)$-approximation \cite{vondrak2008optimal} is the best known algorithm for this problem. This algorithm uses continuous greedy relaxation of the submodular function and randomized rounding. Although we have modeled our objective as the sum of submodular functions, unlike the SWM, we have the same submodular function for each machine; our approach could be viewed as Submodular Welfare Maximization with $b$-matching constraints. In the original SWM problem, there are no constraints on the partition size, but in our problem we are required to set an upper bound on the individual partition sizes. 

\section{Greedy and Lazy Greedy Algorithms}
A popular algorithm for maximizing submodular \bmatching\ is the \greedy \  algorithm \cite{nemhauser1978analysis}, where in each iteration, an edge  with the maximum marginal gain is added to the matching.  In its simplest form
the \greedy \  algorithm could be  expensive to implement, but submodularity
can be exploited to make it efficient. The efficient implementation is known as the  \lgreedy \  algorithm \cite{krause2008near,minoux1978accelerated}. As the maximum gain of each edge decreases monotonically in the course of the algorithm, we can employ a maximum heap to store the gains of the edges. Since the submodular function is normalized, the initial gain of each edge is just the function applied on the edge, 
and at each iteration we pop an edge $e$ from the heap. If $e$ is an \emph{available edge},  i.e., $e$ can be added to the current matching without violating \bmatching \ constraints, we update its marginal gain $g(e)$. We compare $g(e)$ with the next best marginal gain of an edge, available as  the heap's current top.  
If $g(e)$ is greater than or equal to the marginal gain of the current top, we add $e$ to the matching; otherwise we push $e$ to the heap. We iterate on the edges  until the heap becomes empty. Algorithm~\ref{alg:lz} describes the \lgreedy \  approach.

\begin{algorithm}
\caption{\lgreedy\ Algorithm ($G(V,E,W)$)}
\label{alg:lz}
\begin{algorithmic}
\STATE $pq$ = max heap of the edges keyed by marginal gain

\WHILE{$pq$ is not empty} 
    \STATE Edge $e$ = pq.pop()
    \STATE Update marginal gain of $e$
    \IF{ $e$ is available}
        \IF{ marg\_gain of $e  \geq $ marg\_gain of pq.top()}
            \STATE Add $e$ to the matching 
            \STATE update $b(.)$ values of endpoints of $e$ 
        \ELSE
            \STATE push $e$ and its updated gain into $pq$ 
        \ENDIF
    \ENDIF
\ENDWHILE
\end{algorithmic}
\end{algorithm}

The maximum cardinality of a \bmatching \ is bounded by $\beta\, n$. 
In every iteration of the \greedy \  algorithm,  an edge with maximum marginal gain can be chosen in $O(m)$ time. Hence  the time complexity of the \greedy \  algorithm is $O(\beta\, n \,m)$. The  worst-case running time of the \lgreedy \  algorithm is no better than the 
\greedy \  algorithm \cite{minoux1978accelerated}. However, 
by making a reasonable assumption  we can show a better time complexity bound for  the \lgreedy \  algorithm.


The adjacent edges of an edge $e = (u,v)$ constitute the set $N(e) = \{e^\prime: e^\prime \in \delta(u)\ \textrm{or}\ e^\prime \in \delta(v) \}$. Likewise, the adjacent vertices of a vertex $u$ are defined as the set $N(u) = \{v: (u,v) \in \delta(u)\}$.
\begin{Assumption}
\label{assump:local}
The marginal gain of an edge $e$ depends only on its adjacent edges.
\end{Assumption}
With this assumption, when an edge is added to the matching only the marginal gains of adjacent edges change. 
We make this assumption only to analyze the runtime of the algorithms but  not to obtain the quality of the approximation. \smf{This assumption is   applicable to  the objective function in Problem~\ref{sub_b-match_bip} that has been used in  many applications, including the one considered in this paper of load balancing Fock matrix computations.} 
\begin{Lemma}
 Under Assumption~\ref{assump:local}, the time complexity of Algorithm~\ref{alg:lz} is $O(\beta\, m \, \log m)$.
\end{Lemma}
\begin{proof}
The time complexity of Algorithm~\ref{alg:lz} depends on the number of push and pop operations in the max heap. We bound how many times  an edge $e$ is  pushed into the heap. The edge  $e$ is pushed when its updated marginal gain is less than the current top's marginal gain, and thus the number of times the marginal gain of $e$ is updated is an upper bound on the number of push operations on it. From our assumption, the update of the marginal gain of an edge $e$  can happen at most  $2 \beta$
times.  
Hence an edge is pushed into the priority queue O($\beta$) times, and  each of these pushes can take $O(\log m)$ time. Thus the runtime for the all pushes is $O(\beta \, m \log m)$. The number of pop operations are at most the number of pushes. Thus the overall runtime of the \lgreedy\  algorithm for \bmatching \ is $O(\beta\, m \log m)$.
\end{proof}

\section{Locally Dominant Algorithm}
We introduce the concept of $\epsilon$-local dominance, 
use it to design an approximation algorithm for submodular
\bmatching, and prove the correctness of the algorithm.

\subsection{$\epsilon$-Local Dominance and Approximation Ratio} 

The \lgreedy\  algorithm presented in Algorithm~\ref{alg:lz} guarantees a $\frac{1}{3}$ approx. ratio \cite{Calinescu11,Fisher1978} by choosing 
an edge with the  highest marginal gain at each iteration, and thus it is 
 an instance of a globally dominant algorithm. We will show that it is unnecessary to select a globally best edge because the same approximation ratio could be achieved by choosing an edge that is best in its neighborhood.  

Recall that given a matching $M$, an edge  $e$ is \emph{available w.r.t $M$} if both of its end-points are unsaturated in $M$.
\begin{Definition}[Locally dominant matching]
An edge $e$ is \emph{locally dominant} if it is available w.r.t a matching $M$, and the marginal gain of $e$ is greater than or equal to  all available edges adjacent to it.  Similarly,  for an $\epsilon \in (0,1]$, an edge $e$ is  \emph{$\epsilon$- locally dominant} if its marginal gain is at least $\epsilon$ times the marginal gain of any of its available adjacent edges.  A matching $M$ is  \emph{$\epsilon$-locally  dominant} if every  edge of $M$ is  $\epsilon$-locally  dominant when it is added to the matching. 
\end{Definition}

A globally dominant algorithm is also a locally dominant one. Thus our analysis of locally dominant matchings would establish the same approximation ratio for the \greedy\ and \lgreedy\ algorithms.

\begin{theorem}
\label{theorem:one-third}
Any algorithm that produces an $\epsilon$-locally dominant  \bmatching\ is $\frac{\epsilon}{2+\epsilon}$-approximate for a submodular objective function.
\end{theorem}
\begin{proof}
Let $M^*$ denote an optimal matching and $M$ be a matching produced by  an $\epsilon$-locally  dominant algorithm. Denote $|M| =k$. We order the elements of $M$ such that when the edge  $e_i$ is included in  $M$, it is an $\epsilon$-locally  dominant edge. Let $M_i$ denote the locally dominant matching after adding $e_i$ to the set, where $M_0 = \emptyset$ and $M_k = M$.

Our goal is to show that for each  edge in the locally dominant algorithm, we may charge at most two distinct elements of $M^*$. At the $i$th iteration of the algorithm when we add $e_i$ to $M_{i-1}$, we will show that there exists a distinct subset $M_i^* \subset M^*$ with $|M_i^*| \leq 2$ such that $\rho_{e_i}(M_{i-1}) \geq \epsilon\rho_{e^*_j}(M_{i-1})$, for all  ${e^*_j \in M^*_i}$. We will achieve this by maintaining a new sequence of sets $\{T_j\}$, where $T_{i-1}$ is the reservoir of potential edges that  $e_i$ could be charged to. The initial set of this sequence of sets is $T_0$,  which holds the edges in the optimal matching  $M^*$. The sequence of $T$-sets shrink in every iteration 
by removing the elements charged in the previous iteration, so that 
it stores only the candidate elements that could be  charged in this and future iterations. 
 Formally,  $M^*=T_0 \supseteq T_1 \supseteq \dots \supseteq T_k = \emptyset$ such that for $1\leq i \leq k$, the  following two conditions hold.
\newline i) $M_i \cup T_i$ is also a \bmatching\, and 
\newline ii) $M_i \cap T_i = \emptyset$.  
\newline 
The two conditions are satisfied for $M_0$ and $T_0$ because $M_0 \cup T_0 = M^*$ and $M_0 \cap T_0 = \emptyset \cap M^* = \emptyset$.

Now we will describe the charging mechanism at each iteration. We need to construct the reservoir set $T_{i}$  from $T_{i-1}$.  Recall that $e_i$ is added at the $i$th step of the $\epsilon$-locally  dominant matching to obtain $M_i$. There are two cases to consider:
\newline i) If $e_i \in T_{i-1}$, the charging set $M^*_i = \{e_i\}$, 
$M_i = M_{i-1} \cup \{e_i\}$, and $T_i = T_{i-1} \setminus \{e_i\}$. 
\newline ii) Otherwise, let $M^*_i$ be a smallest subset of $T_{i-1}$ such that $(M_{i-1} \cup \ldots  \cup \{e_i\} \cup T_{i-1}) \setminus M^*_i$ is a \bmatching.\ Since a $b$-matching is a 2-extendible system, we know $|M^*_i| \leq 2$. 
Then $M_i = M_{i-1} \cup \{e_i\}$; and $T_i = T_{i-1} \setminus M^*_i$.
\newline Note that the two conditions on $M_i$ and $T_i$ from the previous paragraph are satisfied after these sets are computed
from $M_{i-1}$ and $T_{i-1}$. Since  $M$ 
is a maximal matching, we have $T_k = \emptyset$; otherwise we could have added any of the available edges in $T_k$ to $M$. 

Now when $e_i$ is added to $M_{i-1}$, all the elements of $M^*_i$ are available. This set $M^*_i$  must be the adjacent edges of $e_i$. 
Thus $\forall e^*_j \in M^*_i$, we have $\epsilon \rho_{e^*_j}(M_{i-1}) \leq \rho_{e_i}(M_{i-1})$. We can sum  the inequality for each element of $e^*_j \in M_i^*$, leading to $\sum_{j}\rho_{e^*_j}(M_{i-1}) \leq \frac{2}{\epsilon}\rho_{e_i}(M_{i-1})$.

Rewriting the summation we have,
\begin{align*}
    \rho_{e_i}(M_{i-1}) \geq& \frac{\epsilon}{2}\sum_{j}\rho_{e^*_j}(M_{i-1})\\
                        \geq& \frac{\epsilon}{2} \sum_{j} \rho_{e^*_j}(M_{i-1}\cup\{e^*_1,\dots,e^*_{j-1}\}) \\
                        =& \frac{\epsilon}{2} \sum_{j} (f(M_{i-1} \cup \{e^*_1,\dots,e^*_j\})\\
                        &- f(M_{i-1} \cup \{e^*_1,\dots,e^*_{j-1}\}))\\
                        =& \frac{\epsilon}{2} (f(M_{i-1} \cup \{e^*_1,\dots,e^*_{|M^*_i|}\}) - f(M_{i-1})) \\ 
                        =& \frac{\epsilon}{2} (f(M_{i-1} \cup M^*_i) - f(M_{i-1})) \\
                         \geq& \frac{\epsilon}{2} (f(M \cup M^*_i) - f(M)).
 \end{align*}                        
 \smf{ In line 2, each of the summands is a superset of $M_{i-1}$, and the inequality follows from submodularity of $f$ (Proposition~\ref{prop:monotone}).}  
 Line 3 expresses  the marginal gains in terms of the function $f$. The fourth equality is due to telescoping of the sums,  the fifth equality replaces the set $M^*_i$ for its elements,  and the  last inequality  follows by monotonicity of $f$ (from Proposition~\ref{prop:monotone}).
 
 We now sum over all the elements in $M$ as follows.
\begin{align*}
     \sum_{i} \rho_{e_i}(M_{i-1}) \geq & \frac{\epsilon}{2} \sum_{i} (f(M \cup M^*_i) - f(M)),\\
     f(M) \geq& \frac{\epsilon}{2} \sum_{i} (f(M \cup \{M^*_1 \cup \dots M^*_i\}) \\&- f(M \cup \{M^*_1, \dots, M^*_{i-1}\}))\\
       =& \frac{\epsilon}{2} (f(M \cup M^*) - f(M))\\
          \geq& \frac{\epsilon}{2} (f(M^*) - f(M)). \\
    f(M)       \geq& \frac{\epsilon}{2+\epsilon} f(M^*). 
\end{align*}

 The left side of the second line of the above equations is due to Proposition~\ref{prop:sum}, while the right side comes from Proposition~\ref{prop:monotone}. The next equality telescopes the sum,
 and the fourth inequality is due to monotonicity of $f$. Finally the last line is a restatement of the inequality above it. 
\end{proof}

\begin{corollary}
\label{cor:semi-match}
Any algorithm that produces an $\epsilon$-locally dominant semi-matching  is $\frac{\epsilon}{1+\epsilon}$-approximate for a submodular objective function.
\end{corollary}
\begin{proof}
A semi-matching (there are matching constraints on only one vertex part in a bipartite graph) forms a matroid, which is a 1-extendible system \cite{mestre2006greedy}. So by definition of 1-extendible system, $|M^*_i| \leq 1$. We can substitute this value in appropriate places in the proof of Lemma~\ref{theorem:one-third} and get the desired ratio.
\end{proof}

\subsection{Local Lazy Greedy Algorithm} 
Now we  design a locally dominant edge algorithm to compute a \bmatching,  outlining our approach in Algorithm~\ref{alg:llg}. 
We say that a vertex $v$ is \emph{available} if there is an available edge incident on it, i.e., adding the edge to the matching  does not violate the $b(v)$ constraint. 

For each vertex $v \in V$, we  maintain a priority queue that stores the edges incident on  $v$. The key value of the queue is the marginal gain of the adjacent edges. 
At each iteration of the algorithm we alternate between two operations:  \textit{update} and  \textit{matching}. In the \textit{update} step, we update a  best incident edge of an unmatched vertex $v$. 
Similar to \lgreedy, we can make use of the monotonicity of the marginal gains, and the lazy evaluation process is shown in Algorithm~\ref{alg:lz-eval}. After this step, we can consider a best incident edge for each vertex  as a candidate to be matched. We also maintain an array (say \textit{pointer}) of size $|V|$ that holds the best vertex found in the \textit{update} step. The next step is the actual \textit{matching}. We scan over all the available vertices $v \in V$ and check whether $pointer(v)$ also points to $v$ (i.e., $pointer(pointer(v)) = v$). If this condition is true, we have identified a locally dominant edge, and we add it to the matching. We continue the two steps until no available edge remains. 

\begin{algorithm}
\caption{Lazy Evaluation (Max Heap pq)}
\label{alg:lz-eval}
\begin{algorithmic}[1]
\WHILE{$pq$ is not empty}
    \STATE Edge $e$ = pq.pop()
    \STATE Update  marginal gain of $e$
    \IF{ $e$ is available}
        \IF{ marg\_gain of $e$ $ \geq $ marg\_gain of pq.top()}
            \STATE \textbf{break}
        \ELSE
            \STATE push $e$ and its updated gain into $pq$
        \ENDIF
    \ENDIF
\ENDWHILE
\end{algorithmic}
\end{algorithm}

\begin{algorithm}
\caption{Local Lazy Greedy Algorithm}
\label{alg:llg}
\begin{algorithmic}[1]
\LineComment{Initialization} 
\FOR{$v \in V$}
    \STATE pq(v) $\coloneqq$ max-heap of the incident edges keyed by marginal gain
    \STATE pointer(v) = pq(v).top
\ENDFOR
\STATEx 
\LineComment{Main Loop}
\WHILE{$\exists$ an edge  with  both its endpoints available}
\LineComment{Updating}
\FOR{$v \in V$ such that  $u$ is available}
    \STATE Update pq(v) using  Lazy Evaluation (pq(v))
    \STATE pointer(v) = pq(v).top
\ENDFOR
\STATEx 
\LineComment{Matching }
\FOR{ $u \in V$ such that  $u$ is available}
    \STATE v = pointer(u)
    
    \IF{$v$ is available and pointer(v) == u}
        \STATE M = M $\cup$ $\{u,v\}$
    \ENDIF
\ENDFOR
\ENDWHILE

\end{algorithmic}
\end{algorithm}

\Alex{We omit the short proofs of the following two results.}
\begin{Lemma}
The \llgreedy\  
algorithm is locally dominant.
\end{Lemma}
\begin{corollary}
\Alex{For the \bmatching \ problem with submodular objective,  the \llgreedy \  algorithm is $1/3$-approximate. }
\end{corollary}

\begin{Lemma}
 Under Assumption~\ref{assump:local}, the time complexity of Algorithm~\ref{alg:llg} is $O(\beta\, m \log \Delta)$.
\end{Lemma}
\begin{proof}

As for the  Lazy Greedy algorithm, the number of total push operations is $O(m \beta \,\log{\Delta})$ (the argument of the logarithm  is $\Delta$ instead of $m$ because the maximum size of a priority queue is $\Delta$).    We  maintain two \smf{arrays}, say \textit{PotentialU} and \textit{PotentialM}, of vertices that hold the candidate vertices for iteration in the \emph{update} and \emph{matching} step, respectively. Initially all the vertices are in \emph{PotentialU} and \emph{PotentialM} is empty.  The two arays are set to empty after their corresponding step. In the \emph{update} phase, we insert the vertices for which the marginal gain changed into \emph{PotentialM}. 
In the \emph{matching} step, we iterate only over the vertices in \emph{PotentialM} array. When an edge $(u,v)$ is matched in the \emph{matching} step, we insert $u$, $v$ if they are unsaturated and all their available neighboring vertices into the \emph{PotentialU}. This is the \smf{array} on which in the next iteration, \emph{update} would iterate. Since a vertex $u$ can be inserted at most $b(u) + \sum_{v \in N(u)}b(v)$ times into the array, the overall size of \emph{PotentialU} array during the execution of the algorithm  is $O(m \beta)$. The \emph{PotentialM} is always a subset of \emph{PotentialU}. So it is also bounded by $O(m \beta)$.
Combining all these we get, an $O( \beta \, m \log \Delta)$ 
time complexity.
\end{proof}

\subsection{Parallel Implementaion of Local Lazy Greedy}


\begin{algorithm}[t]
\caption{Parallel Local Lazy Greedy}
\label{alg:pllg}
\begin{algorithmic}[1]
\LineComment{Initialization} 
\FOR{$v \in V$  \textbf{in parallel}}
    \STATE pq(v) $\coloneqq$ max-heap of the incident edges keyed by marginal gain
    \STATE pointer(v) = pq(v).top
\ENDFOR
\STATEx 
\LineComment{Main Loop}
\WHILE{$\exists$ an edge 
where both endpoints 
are available}
\LineComment{Updating}
\FOR{$v \in V$ such that  $v$ is available \textbf{in parallel} }
    \STATE Update pq(v) according to Lazy Evaluation (pq(v))
    \STATE pointer(v) = pq(v).top
\ENDFOR
\STATEx 
\LineComment{Matching }
\FOR{ $u \in V$ such that $u$ is available \textbf{in parallel} }
    \STATE v = pointer(u)
    
    \IF{$v$ is available and  $u < v$ and pointer(v) == u }
            \STATE Mark $(u,v)$ as a matching edge 
    \ENDIF
\ENDFOR
\ENDWHILE

\end{algorithmic}
\end{algorithm}

Both the standard \greedy\ and \lgreedy\ algorithm offer little to no concurrency. The \greedy\ algorithm requires global ordering of the gains after each iteration, and the \lgreedy\ has to maintain a global priority queue. On the other hand, the  \llgreedy\ algorithm is concurrent. Here local dominance is sufficient to maintain the desired approximation ratio. We present a shared memory parallel algorithm based on the serial \llgreedy\ in Algorithm~\ref{alg:pllg}. 

\smf{One key difference between the parallel and the serial algorithms is on maintaining the \emph{potentialU} and \emph{potentialM} arrays. One option is for each of the processors to maintain individual \emph{potentialU} and \emph{potentialM} arrays and concatenate them after the corresponding steps. These arrays may contain duplicate vertices, but they can be handled as follows. We maintain a bit array of size of $n$ initialized to $0$ in each position. This bit array would be reset to $0$ at every iteration. We only process vertices that  have $0$ in its corresponding position in the array. To make sure that only one processor is working on the vertex, we use an atomic {\tt test-and-set} instruction to set the corresponding bit of the array. Thus the total work in the parallel algorithm is the same as of that the serial one i.e., $O( \beta \, m \log \Delta)$. Since the fragment inside the while loop is embarrassingly parallel, the parallel runtime depends on the number of iterations. This number depends on the weights and the edges in the graph, but in the worst case,  could be $O(\beta n)$. We leave it for future work to bound  the number of iterations under different weight distributions (say random) and different graph structures.  }
\section{Experimental Results}

The experiments on the serial algorithm  were run on an Intel Haswell
CPUs  with 2.60 GHz clock speed and 512 GB memory.
The parallel algorithm  was executed on an Intel Knights Landing node 
with a Xeon Phi processor (68 physical cores per node) with 1.4 GHz clock speed and 96 GB DDR4 memory. 

\subsection{Dataset}
We tested our algorithm on both real-world and synthetic graphs shown in Table~\ref{tab:Problems}. (All Tables and Figures from this section are at the end of the paper.)  We generated two classes of RMAT graphs: (a) G500, representing graphs with skewed degree distributions from the  Graph 500 benchmark \cite{graph500} and (b)  SSCA, from the HPCS Scalable Synthetic Compact Applications graph analysis (SSCA\#2) benchmark using the following parameter settings: (a)  $a=0.57$, $b=c= 0.19$, and $d=0.05$ for G500,  and (b) $a=0.6$, and $b=c=d=0.4/3$ for SSCA. 
Moreover, we considered eight  problems taken from the  SuiteSparse Matrix Collection   \cite{FMC11} covering application areas such as medical science, structural engineering, and sensor data. We also included a large web-crawl graph(\textit{eu-2015})~\cite{BMSB} and a movie-interaction network(\textit{hollywood-2011})~\cite{BoVWFI}.

\subsection{Serial Performance}

In Table~\ref{tab:ser} 
we compare the \llgreedy\ algorithm with the \lgreedy\ algorithm.   
Each edge weight is chosen uniformly at  random from the set $[1,5]$. The submodular function employed here is the concave polynomial with $\alpha = 0.5$, and $b=5$ for each vertex. Since both \lgreedy\ and \llgreedy\  algorithms  have equal approximation ratios, the objective function values  computed by them are equal, but  the \llgreedy\  algorithm is faster.  For the largest problem in the dataset, the  \llgreedy\ algorithm is about five times faster than the \lgreedy, \Alex{and it is about three times faster in geometric mean.} 

\subsection{Parallel Performance}
Performance of the parallel implementations of the  \llgreedy\ algorithm is shown  
by a scalibility plot in  
Figure~\ref{fig:shm2}.  
Figure~\ref{fig:shm2} reports results from  a machine with $68$ threads, with all the cores  on a single socket. We see that all problems show good speedups, and   all but three problems 
show good scaling with high numbers of threads. 

\section{Load Balancing in Quantum Chemistry}
\label{sec:lbr}

\label{sec:Fock} 
We show an application of submodular \bmatching\ in   Self-Consistent  Field  (SCF) computations  in computational chemistry~\cite{scf}. 

\subsection{Background}
The SCF calculation is iterative, and we focus on the computationally dominant kernel that is executed in every SCF iteration: the two-electron contribution to the Fock matrix build. The algorithm 
executes forty to fifty iterations to converge to a  predefined tolerance. 

The two-electron contribution involves a $\Theta(n^4)$ calculation over  $\Theta(n^2)$ data elements,  where $n$ is the number of {basis functions}.  The computation is organized as a  set of $n^4$ tasks, where only a small percentage ($< 1\%$) of tasks contribute to the Fock matrix build. Before starting the main SCF iterative loop,  the work required for the Fock matrix  build in each iteration is computed  from the number of nonzeros in the matrix, which is proportional to the work across all SCF iterations. 
This step is inexpensive since it only captures the execution pattern of the Fock matrix build algorithm without performing other computations. The task assignment is recorded prior to the first iteration and then reused across all SCF iterations.

The Fock matrix build itself is also iterative (written as a $\Theta(n^4)$ loop), where each iteration represents a task that computes some elements of the Fock matrix.  For a given iteration, a task is only executed upon satisfying some domain constraints based on the values in two other pre-computed matrices, the  Schwarz and density matrices.

The default load balancing used in NWChemEx~\cite{Apre+:NWChemEx} is to assign iteration indices of the outermost two loops in the Fock matrix build across MPI ranks using an atomic counter based work sharing approach. All MPI ranks atomically increment a global shared counter to identify the loop iterations to execute. This approach limits scalability of the Fock build since the work and number of tasks across MPI ranks are not guaranteed to be balanced.

The task assignment problem here naturally corresponds to a \bmatching\  problem. Let $G(U,V,W)$ be a complete bipartite graph, where $U,V,W$ represent the sets of blocks of the Fock matrix, the set of machines, and the load of the (block,machine) pairs, respectively. The $b$ value for each vertex in $U$ is set to $1$;  \Alex{for each vertex in $V$, it is set to $\lceil|U|/|V|\rceil$ in order to balance the number of MPI messages that each processor needs to send.}  We will show that a submodular objective with these \bmatching\ constraints implicitly encodes the desired load balance. To motivate this, we use the square root function ($\alpha = 0.5$) as our objective function in Eqn.~(\ref{sub_b-match_bip}). 

We consider  the execution of the Greedy algorithm for Submodular \bmatching\ \Alex{on a small example}  
consisting of four tasks with work loads of $300$, $200$, $100$ and $50$
on two machines $M1$ and $M2$. The \bmatching\ constraint requires each processor to be assigned two tasks. 
At the first iteration, we assign the first block (load $300$) to machine $M1$. Note that assigning the second block to machine $M1$ would have the same marginal gain as assigning it  to $M2$ if the objective function were linear. But since  the square root objective function is submodular, the marginal gain of assigning the second block to the second machine is higher than assigning it to the first machine. So we will assign the second block  (load $200$) to machine $M2$. 
Then the third block of work $100$ would be assigned to $M2$ rather
than $M1$, due to the higher marginal gain, and finally the last block with load $50$ would be assigned to $M1$ due to the \bmatching\ constraint. 
We see that modeling the objective by a submodular function implicitly provides the desired load balance, and  the experimental results 
will confirm this.


\subsection{Performance Results}
As a representative bio-molecular system we chose the Ubiquitin protein to test performance,  varying the basis functions used in the computation to represent molecular orbitals, and to demonstrate the capability of our implementation to handle large problem sizes.
The assignment algorithm is general enough to be applied to any scenario where such computational patterns exist, and does not depend on the molecule or the basis functions used. 

We visualize the load  on the processors  in Fig.~\ref{fig:load-bal}. 
The standard deviation for the current assignment is $10^5$, and the coefficient of variation (Std./Avg.) is $7.5 \times 10^{-2}$; 
while these quantities for the submodular assignment are  $436$ and $3
\times 10^{-4}$,  respectively. It is clear that the latter assignment achieves much better load balance than the former. 
The run time is  plotted against the number of processors  in Figure~ 
\ref{fig:scale_nwchem_6-31g}. It can be seen that the current assignment does not scale beyond $3000$ processors, where as the submodular assignment scales to $8000$ processors of Summit. The better load balance also leads to a four-fold speedup over the default assignment. 
Since the Fock matrix computation takes about fifty iterations, we reduce 
the total run time from $30$ minutes to $8$ minutes on Summit. 

\begin{figure}[h]
    \centering
    \includegraphics[width=\columnwidth]{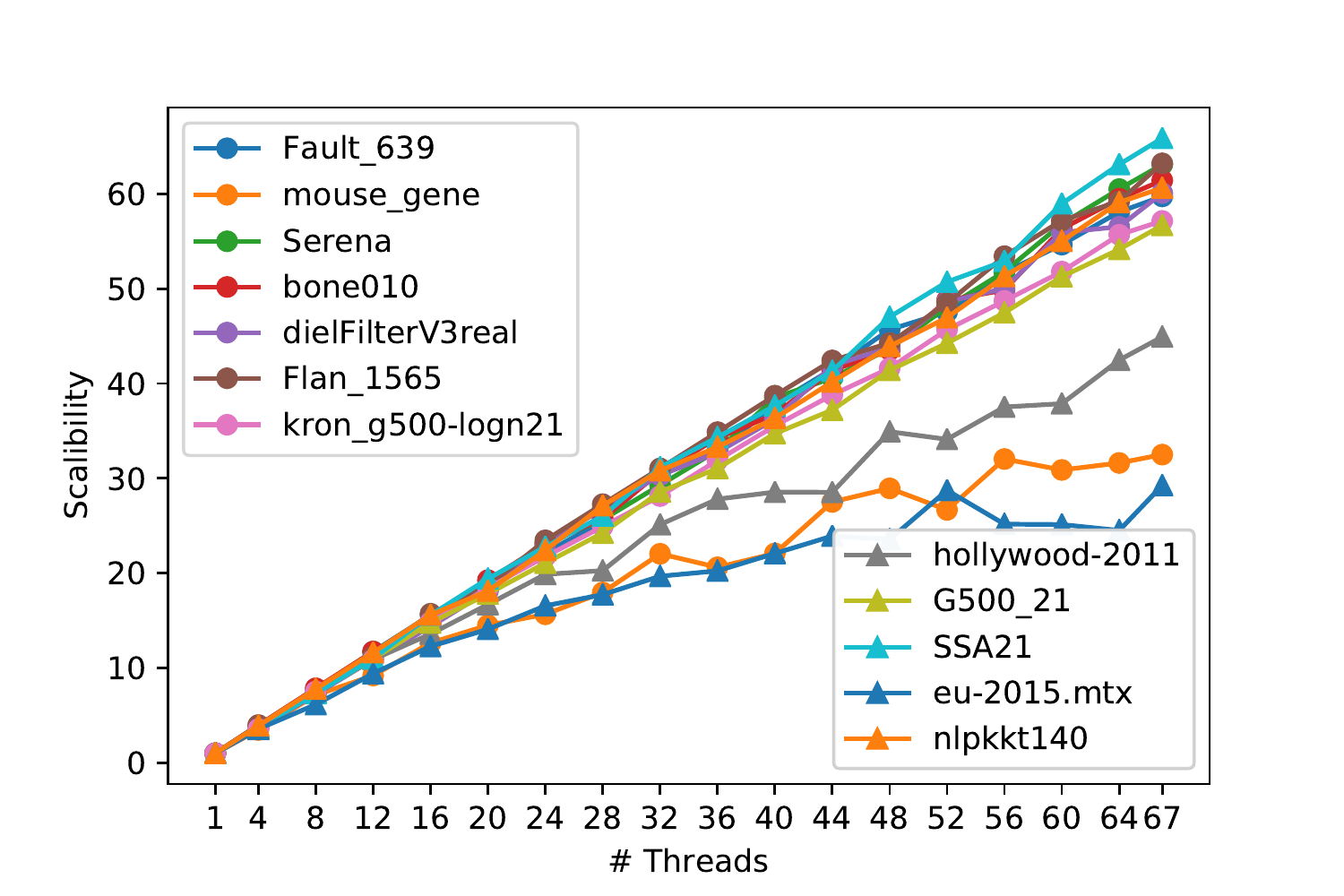}
    \caption{Scalability of the \llgreedy\ algorithm for submodular $b$-matching with  67 threads.}
    \label{fig:shm2}
\end{figure}

\begin{figure}[h]
    \centering
    \includegraphics[width=\columnwidth]{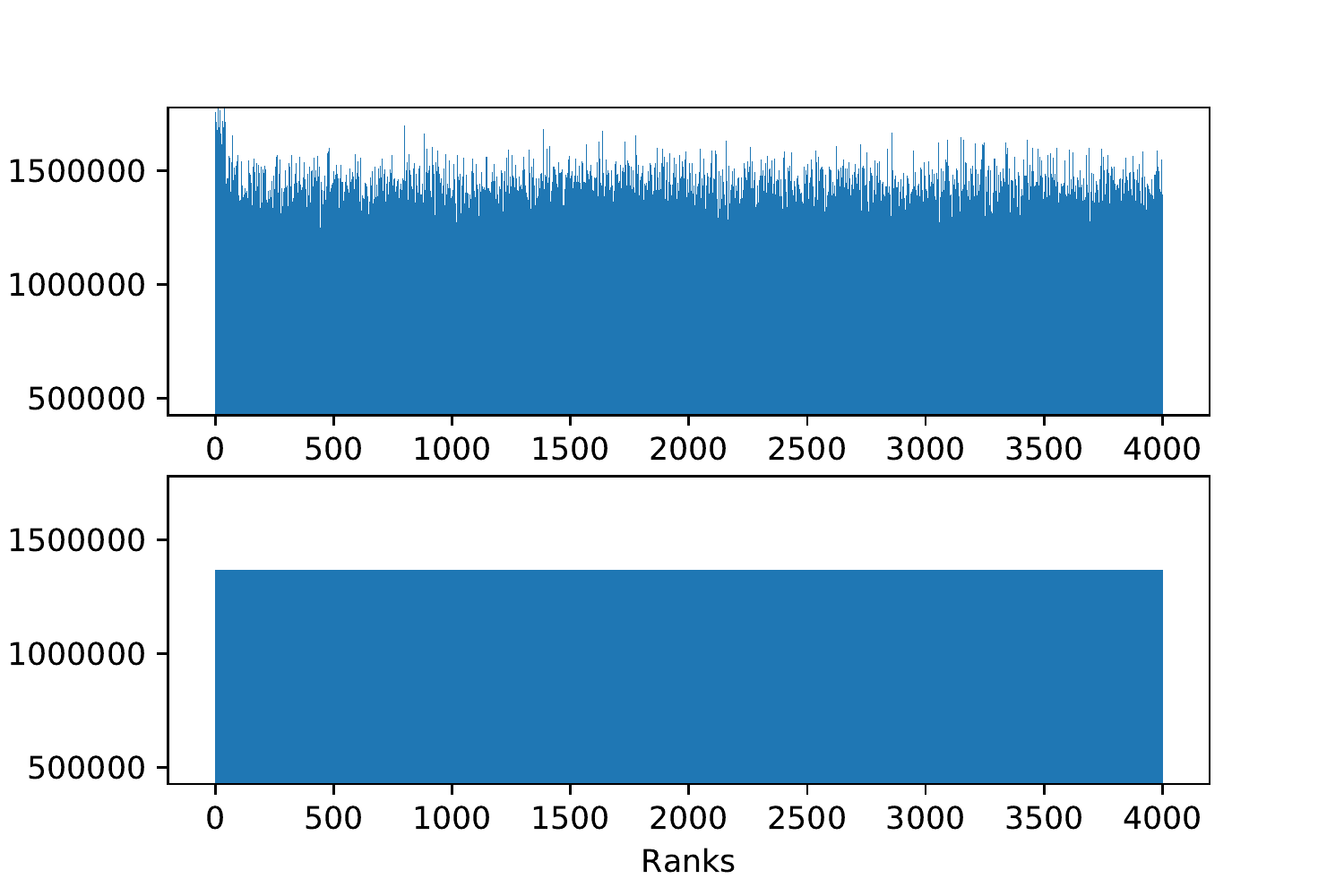}
    \caption{Visualizing the load distribution for the Fock matrix computation for the Ubiquitin protein. Results from:  Top, current assignment  on NWChemEx. Bottom,  submodular assignment.}
    \label{fig:load-bal}
\end{figure}

\begin{figure}[h]
    \centering
    \includegraphics[scale=0.4]{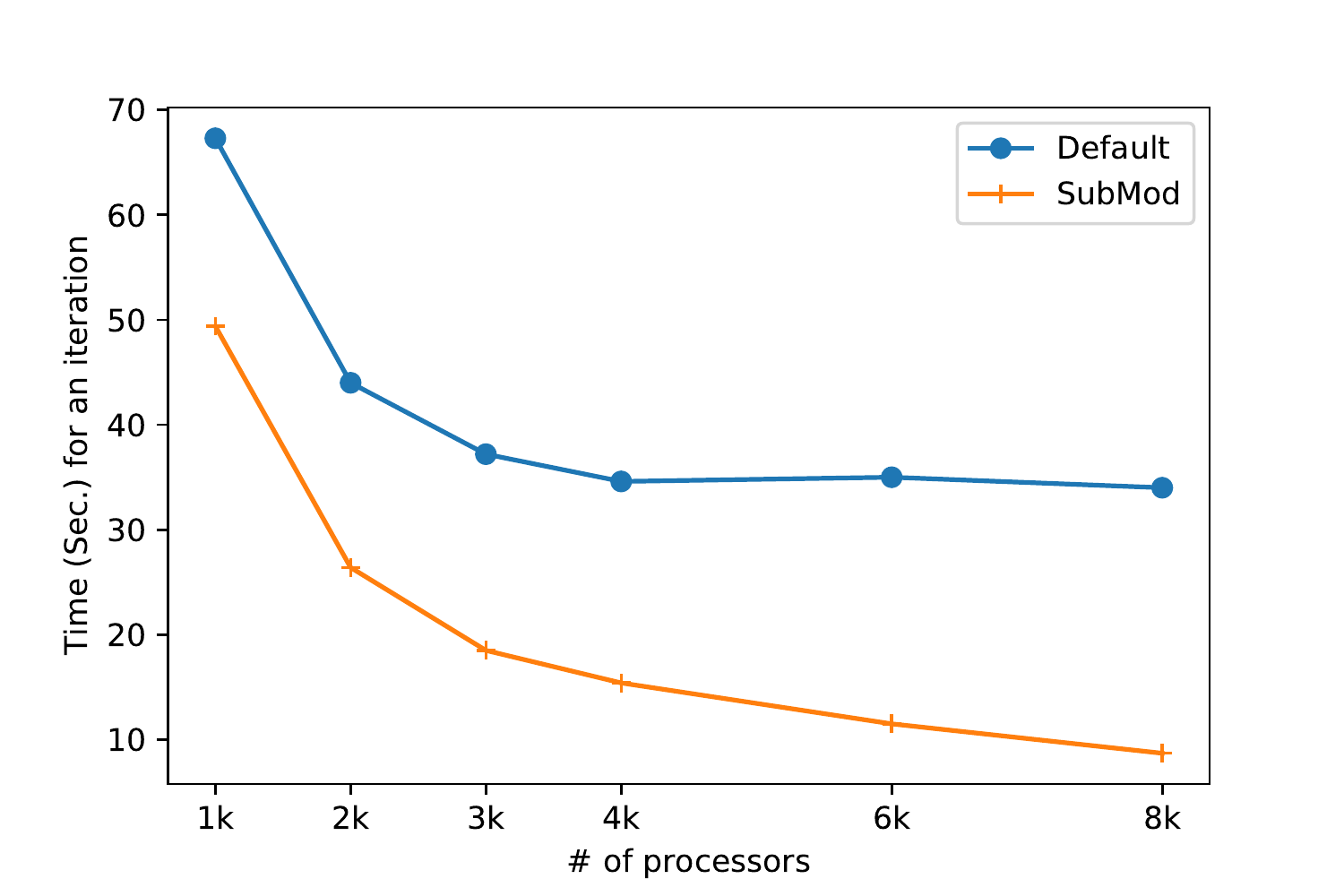}
    \caption{Runtime  per iteration for the current (default) and submodular assignments with the  6-31g basis functions \Alex{for the Ubiquitin protein}  in NWChemEx on
    Summit.}
    \label{fig:scale_nwchem_6-31g}
\end{figure}

\begin{table}[h]
\begin{tabular}{lrrr}
\hline
\textrm{Problems}          & \textrm{Vertices} & \textrm{Edges} & \textrm{Mean} \\
                &         &  & \textrm{Degree} \\
\hline
Fault\_639        & 638,802                      & 13,987,881                & 44                              \\
mouse\_gene       & 45,101                       & 14,461,095                & 641                             \\
Serena            & 1,391,349                    & 31,570,176                & 45                              \\
bone010           & 986,703                      & 35,339,811                & 72                              \\
dielFilterV3real  & 1,102,824                    & 44,101,598                & 80                              \\
Flan\_1565        & 1,564,794                    & 57,920,625                & 74                              \\
kron\_g500-logn21 & 2,097,152                    & 91,040,932                & 87                              \\
hollywood-2011    & 2,180,759                    & 114,492,816               & 105                             \\
G500\_21          & 2,097,150                    & 118,594,475               & 113                             \\
SSA21             & 2,097,152                    & 123,097,397               & 117                             \\
eu-2015           & 11,264,052                   & 257,659,403               & 46                              \\
nlpkkt240         & 27,993,600                   & 373,239,376               & 27                             \\
\hline
\end{tabular}
\caption{The properties of the test graphs listed by increasing number of edges.}
\label{tab:Problems}
\end{table}

\begin{table}[h]
\centering
\begin{tabular}{lrrrrr}
\hline
Problems          & \multicolumn{1}{l}{Weight}                       & \multicolumn{2}{l}{Time (sec.)}                  & \multicolumn{1}{l}{Rel. Perf} \\
\hline
                   & \multicolumn{1}{l}{} & \multicolumn{1}{l}{LG} & \multicolumn{1}{l}{LLG} & \multicolumn{1}{l}{LG/LLG}    \\ \hline
Fault\_639                      & 3.07E+06                & 61.05                  & 16.83                   & 3.63                          \\
mouse\_gene                      & 1.90E+05                & 50.68                  & 22.41                   & 2.26                          \\
Serena                          & 6.69E+06                & 155.81                 & 40.27                   & 3.87                          \\
bone010                          & 4.80E+06                & 177.37                 & 44.15                   & 4.02                          \\
dielFilterV3real                 & 5.35E+06                & 221.92                 & 62.22                   & 3.57                          \\
Flan\_1565                       & 7.63E+06                & 310.31                 & 72.00                   & 4.31                          \\
kron\_g500-logn21                & 3.69E+06                & 304.85                 & 105.58                  & 2.89                          \\
hollywood-2011                   & 8.59E+06                & 622.73                 & 163.26                  & 3.81                          \\
G500\_21                        & 3.93E+06                & 344.13                 & 137.06                  & 2.51                          \\
SSA21                            & 9.46E+06                & 588.16                 & 285.79                  & 2.06                          \\
eu-2015                         & 2.40E+07                & 1098.40                & 396.16                  & 2.77                          \\
nlpkkt240                        & 1.31E+08                & 2456.34                & 465.30                  & 5.28                          \\
\hline
Geo. Mean            &     &    &     & \textbf{3.29}  \\
\hline
\end{tabular}
\caption{The objective function values and comparison of the serial run times  for the  \lgreedy\ and \llgreedy\ algorithms.}
\label{tab:ser}
\end{table}

\clearpage 

\bibliographystyle{plain}
\bibliography{ref.bib}

\begin{thebibliography}{10}

\bibitem{ahmadi2019algorithm}
Saba Ahmadi, Faez Ahmed, John~P. Dickerson, Mark Fuge, and Samir Khuller.
\newblock An algorithm for multi-attribute diverse matching.
\newblock In {\em Proceedings of the Twenty-Ninth International Joint
  Conference on Artificial Intelligence, {IJCAI} 2020}, pages 3--9, 2020.

\bibitem{Bai16}
Wenruo Bai, Jeffrey Bilmes, and William~S Noble.
\newblock Bipartite matching generalizations for peptide identification in
  tandem mass spectrometry.
\newblock In {\em Proceedings of the 7th ACM International Conference on
  Bioinformatics, Computational Biology, and Health Informatics}, pages
  327--336, 2016.

\bibitem{Bai19}
Wenruo Bai, Jeffrey Bilmes, and William~S Noble.
\newblock Submodular generalized matching for peptide identification in tandem
  mass spectrometry.
\newblock {\em IEEE/ACM Transactions on Computational Biology and
  Bioinformatics}, 16(4):1168--1181, 2018.

\bibitem{BMSB}
Paolo Boldi, Andrea Marino, Massimo Santini, and Sebastiano Vigna.
\newblock {BUbiNG}: Massive crawling for the masses.
\newblock In {\em Proceedings of the Companion Publication of the 23rd
  International Conference on World Wide Web}, pages 227--228, 2014.

\bibitem{BoVWFI}
Paolo Boldi and Sebastiano Vigna.
\newblock The {WebGraph} framework {I}: Compression techniques.
\newblock In {\em Proceedings of the 13th International Conference on World
  Wide Web}, pages 595--602, 2004.

\bibitem{buchbinder2018survey}
Niv Buchbinder and Moran Feldman.
\newblock Submodular functions maximization problems.
\newblock In Teofilo~F. Gonzalez, editor, {\em Handbook of Approximation
  Algorithms and Metaheuristics, Second Edition, Volume 1: Methologies and
  Traditional Applications}, pages 753--788. Chapman and Hall/CRC, 2018.

\bibitem{Calinescu11}
Gruia Calinescu, Chandra Chekuri, Martin P{\'a}l, and Jan Vondr{\'a}k.
\newblock Maximizing a monotone submodular function subject to a matroid
  constraint.
\newblock {\em SIAM Journal on Computing}, 40(6):1740--1766, 2011.

\bibitem{chaoji2012recommendations}
Vineet Chaoji, Sayan Ranu, Rajeev Rastogi, and Rushi Bhatt.
\newblock Recommendations to boost content spread in social networks.
\newblock In {\em Proceedings of the 21st International Conference on World
  Wide Web}, pages 529--538, 2012.

\bibitem{chekuri2005polynomial}
Chandra Chekuri and Sanjeev Khanna.
\newblock A polynomial time approximation scheme for the multiple knapsack
  problem.
\newblock {\em SIAM Journal on Computing}, 35(3):713--728, 2005.

\bibitem{FMC11}
Tim Davis and Yifan Hu.
\newblock {The University of Florida Sparse Matrix Collection}.
\newblock {\em ACM Transactions on Mathematical Software}, 38(1):1:1--1:25,
  2011.

\bibitem{Dickerson18}
John~P Dickerson, Karthik~Abinav Sankararaman, Aravind Srinivasan, and Pan Xu.
\newblock Balancing relevance and diversity in online bipartite matching via
  submodularity.
\newblock In {\em Proceedings of the AAAI Conference on Artificial
  Intelligence}, volume~33, pages 1877--1884, 2019.

\bibitem{drake2003simple}
Doratha~E Drake and Stefan Hougardy.
\newblock A simple approximation algorithm for the weighted matching problem.
\newblock {\em Information Processing Letters}, 85(4):211--213, 2003.

\bibitem{Feige98}
Uriel Feige.
\newblock A threshold of $\ln{n}$ for approximating set cover.
\newblock {\em Journal of the ACM (JACM)}, 45(4):634--652, 1998.

\bibitem{Feldman11}
Moran Feldman, Joseph~Seffi Naor, Roy Schwartz, and Justin Ward.
\newblock Improved approximations for $k$-exchange systems.
\newblock In {\em Proceedings of the European Symposium on Algorithms}, pages
  784--798. Springer, 2011.

\bibitem{Fisher1978}
M.~L. Fisher, G.~L. Nemhauser, and L.~A. Wolsey.
\newblock An analysis of approximations for maximizing submodular set
  functions---{II}.
\newblock In M.~L. Balinski and A.~J. Hoffman, editors, {\em Polyhedral
  Combinatorics: Dedicated to the memory of D.R. Fulkerson}, pages 73--87.
  Springer Berlin Heidelberg, Berlin, Heidelberg, 1978.

\bibitem{Fuji16}
Kaito Fujii.
\newblock Faster approximation algorithms for maximizing a monotone submodular
  function subject to a $b$-matching constraint.
\newblock {\em Information Processing Letters}, 116(9):578--584, 2016.

\bibitem{scf}
Tracy~P Hamilton and Henry~F Schaefer~III.
\newblock New variations in two-electron integral evaluation in the context of
  direct {SCF} procedures.
\newblock {\em Chemical Physics}, 150(2):163--171, 1991.

\bibitem{khan2016efficient}
Arif Khan, Alex Pothen, Md~Mostofa Ali~Patwary, Nadathur~Rajagopalan Satish,
  Narayanan Sundaram, Fredrik Manne, Mahantesh Halappanavar, and Pradeep Dubey.
\newblock Efficient approximation algorithms for weighted b-matching.
\newblock {\em SIAM Journal on Scientific Computing}, 38(5):S593--S619, 2016.

\bibitem{Apre+:NWChemEx}
Karol Kowalski, Raymond Bair, Nicholas~P. Bauman, Jeffery~S. Boschen, Eric~J.
  Bylaska, Jeff Daily, Wibe~A. de~Jong, Thom Dunning, Niranjan Govind,
  Robert~J. Harrison, Murat Keçeli, Kristopher Keipert, Sriram Krishnamoorthy,
  Suraj Kumar, Erdal Mutlu, Bruce Palmer, Ajay Panyala, Bo~Peng, Ryan~M.
  Richard, T.~P. Straatsma, Peter Sushko, Edward~F. Valeev, Marat Valiev,
  Hubertus J.~J. van Dam, Jonathan~M. Waldrop, David~B. Williams-Young, Chao
  Yang, Marcin Zalewski, and Theresa~L. Windus.
\newblock From {NWChem} to {NWChemEx}: Evolving with the computational
  chemistry landscape.
\newblock {\em Chemical Reviews}, 121(8):4962--4998, 2021.
\newblock PMID: 33788546.

\bibitem{Krause14}
Andreas Krause and Daniel Golovin.
\newblock Submodular function maximization.
\newblock In Lucas Bordeaux, Youssef Hamadi, and Pushmeet Kohli, editors, {\em
  Tractability: Practical Approaches to Hard Problems}, pages 71--104.
  Cambridge University Press, 2014.

\bibitem{Krause05}
Andreas Krause and Carlos Guestrin.
\newblock Near-optimal nonmyopic value of information in graphical models.
\newblock In {\em Proceedings of the Twenty-first Conference on Uncertainty in
  Artificial Intelligence}, pages 324--331, 2005.

\bibitem{krause2008near}
Andreas Krause, Ajit Singh, and Carlos Guestrin.
\newblock Near-optimal sensor placements in {G}aussian processes: Theory,
  efficient algorithms and empirical studies.
\newblock {\em Journal of Machine Learning Research}, 9(8):235--284, 2008.

\bibitem{lehmann2006combinatorial}
Benny Lehmann, Daniel Lehmann, and Noam Nisan.
\newblock Combinatorial auctions with decreasing marginal utilities.
\newblock {\em Games and Economic Behavior}, 55(2):270--296, 2006.

\bibitem{lenstra1990approximation}
Jan~Karel Lenstra, David~B Shmoys, and {\'E}va Tardos.
\newblock Approximation algorithms for scheduling unrelated parallel machines.
\newblock {\em Mathematical Programming}, 46(1):259--271, 1990.

\bibitem{Lin11}
Hui Lin and Jeff Bilmes.
\newblock Word alignment via submodular maximization over matroids.
\newblock In {\em Proceedings of the 49th Annual Meeting of the Association for
  Computational Linguistics: Human Language Technologies}, pages 170--175,
  2011.

\bibitem{manne2014new}
Fredrik Manne and Mahantesh Halappanavar.
\newblock New effective multithreaded matching algorithms.
\newblock In {\em Proceedings of the 28th International Parallel and
  Distributed Processing Symposium}, pages 519--528. IEEE, 2014.

\bibitem{mestre2006greedy}
Juli{\'a}n Mestre.
\newblock Greedy in approximation algorithms.
\newblock In {\em Proceedings of the 14th European Symposium on Algorithms},
  pages 528--539. Springer, 2006.

\bibitem{minoux1978accelerated}
Michel Minoux.
\newblock Accelerated greedy algorithms for maximizing submodular set
  functions.
\newblock In {\em Proceedings of the 8th IFIP Conference on Optimization
  Techniques}, pages 234--243. Springer, 1977.

\bibitem{graph500}
Richard~C Murphy, Kyle~B Wheeler, Brian~W Barrett, and James~A Ang.
\newblock {Introducing the Graph 500}.
\newblock {\em Cray User's Group}, 2010.

\bibitem{nemhauser1978best}
George~L Nemhauser and Laurence~A Wolsey.
\newblock Best algorithms for approximating the maximum of a submodular set
  function.
\newblock {\em Mathematics of Operations Research}, 3(3):177--188, 1978.

\bibitem{nemhauser1978analysis}
George~L Nemhauser, Laurence~A Wolsey, and Marshall~L Fisher.
\newblock An analysis of approximations for maximizing submodular set
  functions—{I}.
\newblock {\em Mathematical Programming}, 14(1):265--294, 1978.

\bibitem{pettie2004simpler}
Seth Pettie and Peter Sanders.
\newblock A simpler linear time 2/3-$\varepsilon$ approximation for maximum
  weight matching.
\newblock {\em Information Processing Letters}, 91(6):271--276, 2004.

\bibitem{pothen2019approximation}
Alex Pothen, S~M Ferdous, and Fredrik Manne.
\newblock Approximation algorithms in combinatorial scientific computing.
\newblock {\em Acta Numerica}, 28:541--633, 2019.

\bibitem{preis1999linear}
Robert Preis.
\newblock Linear time 1/2-approximation algorithm for maximum weighted matching
  in general graphs.
\newblock In {\em Proceedings of the 16th Annual Conference on Theoretical
  Aspects of Computer Science}, pages 259--269, 1999.

\bibitem{shmoys1993approximation}
David~B Shmoys and {\'E}va Tardos.
\newblock An approximation algorithm for the generalized assignment problem.
\newblock {\em Mathematical Programming}, 62(1):461--474, 1993.

\bibitem{tohidi2020submodularity}
Ehsan Tohidi, Rouhollah Amiri, Mario Coutino, David Gesbert, Geert Leus, and
  Amin Karbasi.
\newblock Submodularity in action: From machine learning to signal processing
  applications.
\newblock {\em IEEE Signal Processing Magazine}, 37(5):120--133, 2020.

\bibitem{vondrak2008optimal}
Jan Vondr{\'a}k.
\newblock Optimal approximation for the submodular welfare problem in the value
  oracle model.
\newblock In {\em Proceedings of the Fortieth Annual ACM Symposium on Theory of
  Computing}, pages 67--74, 2008.

\end{thebibliography}


\end{document}